\documentclass{article}
\usepackage[utf8]{inputenc}
\usepackage{mathrsfs}
\usepackage{amsmath}
\usepackage{amssymb}
\usepackage{amsthm}
\usepackage{authblk}
\usepackage{tikz}
\usetikzlibrary{automata, positioning}
\usepackage{float}
\usepackage{tabularx,lipsum,environ}
\usepackage{algpseudocode}
\usepackage{mathtools}

\usepackage{hyperref}
\newtheorem{observation}{Observation}
\newtheorem{theorem}{Theorem}
\newtheorem{lemma}{Lemma}

\newtheorem{fact}{Fact}
\newtheorem{claim}{Claim}

\usepackage[T1]{fontenc}
\title{Careful Synchronization of One-Cluster Automata}
\author{Jakub Ruszil}
\date{July 2023}

\begin{document}

\maketitle

\begin{abstract}
    In this paper we investigate careful synchronization of one-cluster partial automata. First we prove that in general case the shortest carefully synchronizing word for such automata is of length $2^\frac{n}{2} + 1$, where $n$ is the number of states of an automaton. Additionally we prove that checking whether a given one-cluster partial automaton is carefully synchronizing is NP-hard even in the case of binary alphabet.
\end{abstract}

\section{Introduction}
The concept of synchronization of finite automata is essential in various areas of computer science. It consists in regaining control over a system by applying a specific set of input instructions. These instructions lead the system to a fixed state no matter in which state it was at the beginning. The idea of synchronization has been studied for many classes of complete deterministic finite automata (DFA) \cite{berlinkov2014,berlinkov2016,eppstein1990,kari2001,kari2003,pin1983,rystsov1997,szykula2018,cerny1964,trahtman2007,volkov2008,volkov2019} and non-deterministic finite automata \cite{imreh1999,ito2004}. One of the most famous longstanding open problems in automata theory, known as \v{C}ern\'{y} Conjecture, states that for a given synchronizing DFA with $n$ states one can always find a synchronizing word of length at most $(n-1)^2$. This conjecture was proven for numerous classes of automata, but the problem is still not solved in general case. The concept of synchronization has been also considered in coding theory \cite{biskup2009,jurgensen2008}, parts orienting in manufacturing \cite{eppstein1990,natarajan1986}, testing of reactive systems \cite{sandberg2005} and Markov Decision Processes \cite{doyen2014,doyen2019}.

Allowing no outgoing transitions from some states for certain letters helps us to model a system for which certain actions cannot be accomplished while being in a specified state. This leads to the problem of finding a synchronizing word for a finite automaton, where transition function is not defined for all states. Notice that this is the most frequent case, if we use automata to model real-world systems. In practice, it rarely happens that a real system can be modeled with a DFA where transition function is total. The transition function is usually a partial one. This fact motivated many researchers to investigate the properties of partial finite automata relevant to practical problems of synchronization.\newline
We know that, in general case, checking if a partial automaton can be carefully synchronized is PSPACE-complete \cite{martyugin2010b} even for binary alphabet \cite{vorel2016}.
In this paper we investigate the case of deterministic finite automata such that transition from state to state is not necessary defined for all states. We refer to this model as \textit{partial finite automaton} (PFA). We will say that a word which synchronizes a PFA is \textit{carefully synchronizing} word. There exists also different definition of synchronization of PFAs, for example in here \cite{berlinkov_et_al:LIPIcs.STACS.2021.12}. The problem of estimating the length of a shortest carefully synchronizing word for PFA was considered, among others, by Rystsov \cite{rystsov1980}, Ito and Shikishima-Tsuji \cite{ito2004}, Martyugin \cite{martyugin2010,martyugin2013}, Gazdag et. al. \cite{gazdag2009} and de Bondt et. al. \cite{debondt2019}. Martyugin established a lower bound on the length of such words of size $\Theta(3^{n/3})$ and Rystsov \cite{rystsov1980} established an upper bound of size $O((3+ \epsilon)^{n/3})$, where $n$ is the number of states of the automaton.\newline
One-cluster DFAs are the ones having a letter that has only one simple cycle on the set of states. Synchronization of one-cluster DFA's was studied by Steinberg, B{\'e}al and Perrin and resulted in proving \v{C}ern\'{y} Conjecture for one-cluster automata with prime length cycle in \cite{STEINBERG20115487} and establishing quadratic upper bound for the length of the synchronizing word in general case in \cite{10.1007/978-3-642-02737-6_6} using linear algebra approach. However we also know that in general case deciding if an automaton is synchronizing is in P what is not true in the case of careful synchronization. Following that researches we extend the notion of one-cluster to partial automata. Our research is motivated by the fact, that all former constructions providing high careful synchronization threshold does not contain a letter which creates a cluster on the set of states. Constructions from \cite{martyugin2010,martyugin2013} have linear alphabet size in terms of cardinality of the set of states but only one letter is defined for all states and the number of clusters grows also linearly in the size of the set of states for that letter. Cases from \cite{vorel2016} and \cite{debondt2019} are more interesting as they are dealing with constant alphabet size, but constructions from that papers also do not define any letter to be a one cluster (one of families of automata from \cite{debondt2019}, used in the proof of Lemma 7 has one letter defined as a permutation on the set of states with three cycles). On the other hand binary partial one cluster automata have simple structural condition concerning non-cluster letter (stated in latter Observation \ref{observation:cluster_defined}) which gives the impression that deciding careful synchronizability of such automata is an easy problem in terms of complexity theory. Suprisingly we managed to show that this decision problem is hard even in such restricted case.  Our contribution is twofold. First we prove that in general case the shortest carefully synchronizing words for one-cluster partial automata can be of exponential size in terms of the number of states (but assuming alphabet is also of exponential size). From the complexity perspective we show that the problem of deciding if a given binary one-cluster PFA admits any carefully synchronizing word is NP-hard by reduction from 3-SAT problem.
\section{Preliminaries}
A \emph{partial finite automaton} (PFA) is an ordered tuple $\mathcal{A} = (\Sigma, Q, \delta)$ where $\Sigma$ is a finite set of letters, $Q$ is a finite set of states and $\delta:{Q \times \Sigma}\rightarrow{Q}$ is a transition function, not everywhere defined. For $\emph{w} \in \Sigma^\ast$ and $\emph{q} \in Q$ we define $\delta(\emph{q},\emph{w})$ inductively: $\delta(\emph{q},\epsilon) = q$ and $\delta(\emph{q},\emph{aw}) = \delta(\delta(\emph{q},\emph{a}), \emph{w})$ for $a \in \Sigma$, where $\epsilon$ is the empty word and $\delta(\emph{q}, \emph{a})$ is defined. Let $S \subseteq Q$. Define $\delta(S,w) = \bigcup_{q \in S} \delta(q,w)$. We write $q.w$ (respectively $S.w$ for a image of a set $S \subseteq Q$) instead of $\delta(q,w)$ (respectively $\delta(S,w)$) wherever it does not cause ambiguity. We also denote $S.w^{-1} = \{q \in Q: q.w \in S\}$ as a preimage of $S$ under the word $w$. A word $\emph{w} \in \Sigma^\ast$ is called \emph{carefully synchronizing} if there exists $\overline{q} \in Q$ such that for every $\emph{q} \in Q$, $\delta(\emph{q}, \emph{w}) = \overline{q}$ and all transitions are defined. A PFA is called \emph{carefully synchronizing} if it admits any carefully synchronizing word. For a given $\mathcal{A} = (Q, \Sigma, \delta)$ we define the \emph{power-automaton} $\mathcal{P}(\mathcal{A}) = (2^Q, \Sigma, \tau)$, where $2^Q$ stands for the set of all subsets of $Q$. Transition function $\tau:{2^Q \times \Sigma}\rightarrow{2^Q}$ is defined as follows: $\tau(Q',a) = \bigcup_{q \in Q'} \delta(q,a)$, $a \in \Sigma$, $Q' \subseteq Q$  if $\delta(q,a)$ is defined for all states in $q \in Q'$. Otherwise $\tau(Q',a)$ is not defined. We also consider a \textit{deterministic finite automaton} DFA. The only difference comparing to a PFA is that the transition function is total in this case. All definitions regarding PFA also apply to DFA, but we speak rather about \textit{synchronization} than \textit{careful synchronization} in the case of DFA. We can now state an obvious fact, useful to decide whether a given PFA is carefully synchronizing.
\begin{fact}
\label{fact:1}
Let $\mathcal{A}$ be a PFA and $\mathcal{P(A)}$ be its power automaton. Then $\mathcal{A}$ is synchronizing if and only if for some state $q \in Q$  there exists a labelled path in $\mathcal{P(A)}$ from $Q$ to $\{q\}$. The shortest synchronizing word for $\mathcal{A}$ corresponds to the shortest such labelled path in $\mathcal{P(A)}$. 
\end{fact}
An example of the carefully synchronizing automaton $\mathcal{A}_{car}$ is depicted in Fig. \ref{fig:example}. Its shortest carefully synchronizing word $w_{car}$ is $abc(ab)^2 c^2a$, which can be easily checked via the power automaton construction.
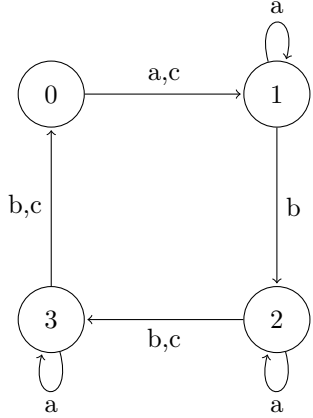
\begin{figure}[H]
    \centering
    \begin{tikzpicture}[shorten >=1pt,node distance=3cm,on     grid,auto] 
        \node[state] (0)   {$0$}; 
        \node[state] (1) [right=of 0] {$1$}; 
        \node[state] (2) [below=of 1] {$2$}; 
        \node[state] (3) [left=of 2] {$3$};
        \path[->] 
        (0) edge node {a,c} (1)
        (1) edge [loop above] node {a} ()
            edge node {b} (2)
        (2) edge [loop below] node {a} ()
            edge node {b,c} (3)
        (3) edge [loop below] node {a} ()
            edge node {b,c} (0);

        \end{tikzpicture}
    \caption{A carefully synchronizing $\mathcal{A}_{car}$}
    \label{fig:example}
\end{figure}
Let $\mathcal{L}_n = \{\mathcal{A} = (\Sigma, Q, \delta): \mathcal{A}\;is\;carefully\;synchronizing\;and\;|Q| = n\}$. Notice that $\mathcal{L}_n$ does not depend on alphabet size. We define $d(\mathcal{A}) = \min\{|w|:w\; is\; a\; carefully\; synchronizing\:word\; for\; \mathcal{A}\}$ and $d(n) = \max\{d(\mathcal{A}) : \mathcal{A} \in \mathcal{L}_n\}$.\newline
It can be easily verified from Fig \ref{fig:example}. that the \v{C}ern\'y Conjecture is not true for PFAs, since $d(\mathcal{A}_{car}) = 10 > (4-1)^2 = 9$.\newline
For $\mathcal{A} = (Q, \Sigma, \delta)$ and $a \in \Sigma$ defined for all the states let $\mathcal{G}_a = (Q, E)$ be a directed graph such that $(p,q) \in E$ if, and only if $p.a = q$. Put differently $\mathcal{G}_a$ is a digraph made of only the edges labelled by $a$ in $\mathcal{A}$. The graph $\mathcal{G}_a$ is a disjoint union of weakly connected components called $a$-\textit{clusters}. Since each state has only one outgoing edge in $\mathcal{G}_a$, each $a$-cluster contains a unique cycle, called $a$-\textit{cycle} with trees attached to a cycle at their roots. For each $p \in Q$ we define a \textit{level} of $p$ as a distance between $p$ and the root of a tree containing $p$ in $\mathcal{G}_a$. The \textit{level} of $\mathcal{G}_a$ is the maximal level of its vertices. An example of the $a$-cluster is shown in Figure \ref{fig:cluster}.
One can easily deduce from Fact \ref{fact:1} that at least one letter $a \in \Sigma$ must be defined for all the states of the PFA in order for it to be carefully synchronizing. We define a \textit{one-cluster} PFA with respect to a letter $a$ as a PFA which has only one $a$-cluster. We also refer to a \textit{one-cluster} PFA as an automaton which is one-cluster with respect to any of its letters.
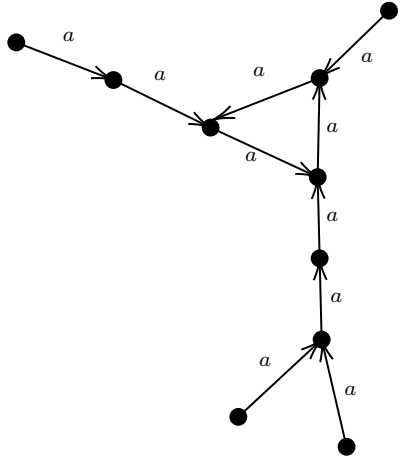
\begin{figure}[H]
    \centering
    \tikzset{every picture/.style={line width=0.75pt}} 

\begin{tikzpicture}[x=0.75pt,y=0.75pt,yscale=-1,xscale=1]

\draw  [fill={rgb, 255:red, 0; green, 0; blue, 0 }  ,fill opacity=1 ] (208,126) .. controls (208,123.79) and (209.79,122) .. (212,122) .. controls (214.21,122) and (216,123.79) .. (216,126) .. controls (216,128.21) and (214.21,130) .. (212,130) .. controls (209.79,130) and (208,128.21) .. (208,126) -- cycle ;
\draw  [fill={rgb, 255:red, 0; green, 0; blue, 0 }  ,fill opacity=1 ] (271,296) .. controls (271,293.79) and (272.79,292) .. (275,292) .. controls (277.21,292) and (279,293.79) .. (279,296) .. controls (279,298.21) and (277.21,300) .. (275,300) .. controls (272.79,300) and (271,298.21) .. (271,296) -- cycle ;
\draw  [fill={rgb, 255:red, 0; green, 0; blue, 0 }  ,fill opacity=1 ] (259.95,153.86) .. controls (257.82,153.28) and (256.56,151.08) .. (257.14,148.95) .. controls (257.72,146.82) and (259.92,145.56) .. (262.05,146.14) .. controls (264.18,146.72) and (265.44,148.92) .. (264.86,151.05) .. controls (264.28,153.18) and (262.08,154.44) .. (259.95,153.86) -- cycle ;
\draw  [fill={rgb, 255:red, 0; green, 0; blue, 0 }  ,fill opacity=1 ] (312,216) .. controls (312,213.79) and (313.79,212) .. (316,212) .. controls (318.21,212) and (320,213.79) .. (320,216) .. controls (320,218.21) and (318.21,220) .. (316,220) .. controls (313.79,220) and (312,218.21) .. (312,216) -- cycle ;
\draw  [fill={rgb, 255:red, 0; green, 0; blue, 0 }  ,fill opacity=1 ] (313,257) .. controls (313,254.79) and (314.79,253) .. (317,253) .. controls (319.21,253) and (321,254.79) .. (321,257) .. controls (321,259.21) and (319.21,261) .. (317,261) .. controls (314.79,261) and (313,259.21) .. (313,257) -- cycle ;
\draw  [fill={rgb, 255:red, 0; green, 0; blue, 0 }  ,fill opacity=1 ] (311,175) .. controls (311,172.79) and (312.79,171) .. (315,171) .. controls (317.21,171) and (319,172.79) .. (319,175) .. controls (319,177.21) and (317.21,179) .. (315,179) .. controls (312.79,179) and (311,177.21) .. (311,175) -- cycle ;
\draw  [fill={rgb, 255:red, 0; green, 0; blue, 0 }  ,fill opacity=1 ] (325.57,311.11) .. controls (325.57,308.9) and (327.36,307.11) .. (329.57,307.11) .. controls (331.78,307.11) and (333.57,308.9) .. (333.57,311.11) .. controls (333.57,313.32) and (331.78,315.11) .. (329.57,315.11) .. controls (327.36,315.11) and (325.57,313.32) .. (325.57,311.11) -- cycle ;
\draw  [fill={rgb, 255:red, 0; green, 0; blue, 0 }  ,fill opacity=1 ] (312,125) .. controls (312,122.79) and (313.79,121) .. (316,121) .. controls (318.21,121) and (320,122.79) .. (320,125) .. controls (320,127.21) and (318.21,129) .. (316,129) .. controls (313.79,129) and (312,127.21) .. (312,125) -- cycle ;
\draw    (261,150) -- (313.19,174.16) ;
\draw [shift={(315,175)}, rotate = 204.84] [color={rgb, 255:red, 0; green, 0; blue, 0 }  ][line width=0.75]    (10.93,-3.29) .. controls (6.95,-1.4) and (3.31,-0.3) .. (0,0) .. controls (3.31,0.3) and (6.95,1.4) .. (10.93,3.29)   ;
\draw  [fill={rgb, 255:red, 0; green, 0; blue, 0 }  ,fill opacity=1 ] (159,107) .. controls (159,104.79) and (160.79,103) .. (163,103) .. controls (165.21,103) and (167,104.79) .. (167,107) .. controls (167,109.21) and (165.21,111) .. (163,111) .. controls (160.79,111) and (159,109.21) .. (159,107) -- cycle ;
\draw    (315,175) -- (315.96,127) ;
\draw [shift={(316,125)}, rotate = 91.15] [color={rgb, 255:red, 0; green, 0; blue, 0 }  ][line width=0.75]    (10.93,-3.29) .. controls (6.95,-1.4) and (3.31,-0.3) .. (0,0) .. controls (3.31,0.3) and (6.95,1.4) .. (10.93,3.29)   ;
\draw    (316,125) -- (263.91,145.41) ;
\draw [shift={(262.05,146.14)}, rotate = 338.6] [color={rgb, 255:red, 0; green, 0; blue, 0 }  ][line width=0.75]    (10.93,-3.29) .. controls (6.95,-1.4) and (3.31,-0.3) .. (0,0) .. controls (3.31,0.3) and (6.95,1.4) .. (10.93,3.29)   ;
\draw    (316,216) -- (315.05,177) ;
\draw [shift={(315,175)}, rotate = 88.6] [color={rgb, 255:red, 0; green, 0; blue, 0 }  ][line width=0.75]    (10.93,-3.29) .. controls (6.95,-1.4) and (3.31,-0.3) .. (0,0) .. controls (3.31,0.3) and (6.95,1.4) .. (10.93,3.29)   ;
\draw    (317,257) -- (316.05,218) ;
\draw [shift={(316,216)}, rotate = 88.6] [color={rgb, 255:red, 0; green, 0; blue, 0 }  ][line width=0.75]    (10.93,-3.29) .. controls (6.95,-1.4) and (3.31,-0.3) .. (0,0) .. controls (3.31,0.3) and (6.95,1.4) .. (10.93,3.29)   ;
\draw    (329.57,311.11) -- (317.45,258.95) ;
\draw [shift={(317,257)}, rotate = 76.93] [color={rgb, 255:red, 0; green, 0; blue, 0 }  ][line width=0.75]    (10.93,-3.29) .. controls (6.95,-1.4) and (3.31,-0.3) .. (0,0) .. controls (3.31,0.3) and (6.95,1.4) .. (10.93,3.29)   ;
\draw    (275,296) -- (315.53,258.36) ;
\draw [shift={(317,257)}, rotate = 137.12] [color={rgb, 255:red, 0; green, 0; blue, 0 }  ][line width=0.75]    (10.93,-3.29) .. controls (6.95,-1.4) and (3.31,-0.3) .. (0,0) .. controls (3.31,0.3) and (6.95,1.4) .. (10.93,3.29)   ;
\draw    (212,126) -- (259.2,149.12) ;
\draw [shift={(261,150)}, rotate = 206.1] [color={rgb, 255:red, 0; green, 0; blue, 0 }  ][line width=0.75]    (10.93,-3.29) .. controls (6.95,-1.4) and (3.31,-0.3) .. (0,0) .. controls (3.31,0.3) and (6.95,1.4) .. (10.93,3.29)   ;
\draw    (163,107) -- (210.14,125.28) ;
\draw [shift={(212,126)}, rotate = 201.19] [color={rgb, 255:red, 0; green, 0; blue, 0 }  ][line width=0.75]    (10.93,-3.29) .. controls (6.95,-1.4) and (3.31,-0.3) .. (0,0) .. controls (3.31,0.3) and (6.95,1.4) .. (10.93,3.29)   ;
\draw  [fill={rgb, 255:red, 0; green, 0; blue, 0 }  ,fill opacity=1 ] (347,91) .. controls (347,88.79) and (348.79,87) .. (351,87) .. controls (353.21,87) and (355,88.79) .. (355,91) .. controls (355,93.21) and (353.21,95) .. (351,95) .. controls (348.79,95) and (347,93.21) .. (347,91) -- cycle ;
\draw    (351,91) -- (317.43,123.61) ;
\draw [shift={(316,125)}, rotate = 315.83] [color={rgb, 255:red, 0; green, 0; blue, 0 }  ][line width=0.75]    (10.93,-3.29) .. controls (6.95,-1.4) and (3.31,-0.3) .. (0,0) .. controls (3.31,0.3) and (6.95,1.4) .. (10.93,3.29)   ;

\draw (185,100) node [anchor=north west][inner sep=0.75pt]  [font=\footnotesize] [align=left] {$\displaystyle a$};
\draw (231,120) node [anchor=north west][inner sep=0.75pt]  [font=\footnotesize] [align=left] {$\displaystyle a$};
\draw (281,118) node [anchor=north west][inner sep=0.75pt]  [font=\footnotesize] [align=left] {$\displaystyle a$};
\draw (277,161) node [anchor=north west][inner sep=0.75pt]  [font=\footnotesize] [align=left] {$\displaystyle a$};
\draw (318,146) node [anchor=north west][inner sep=0.75pt]  [font=\footnotesize] [align=left] {$\displaystyle a$};
\draw (318,191) node [anchor=north west][inner sep=0.75pt]  [font=\footnotesize] [align=left] {$\displaystyle a$};
\draw (320,232) node [anchor=north west][inner sep=0.75pt]  [font=\footnotesize] [align=left] {$\displaystyle a$};
\draw (284,264) node [anchor=north west][inner sep=0.75pt]  [font=\footnotesize] [align=left] {$\displaystyle a$};
\draw (327,279) node [anchor=north west][inner sep=0.75pt]  [font=\footnotesize] [align=left] {$\displaystyle a$};
\draw (335.5,111) node [anchor=north west][inner sep=0.75pt]  [font=\footnotesize] [align=left] {$\displaystyle a$};

\end{tikzpicture}
    \caption{Example of the $a$-cluster}
    \label{fig:cluster}
\end{figure}
Let $\mathcal{A} = (Q, \Sigma, \delta)$ be a PFA, $a \in \Sigma$, $S \subset Q$ and $G = (V,E)$ be a digraph. We say that \emph{$S$ induces $G$ on the letter $a$ in $\mathcal{A}$}, if a set $S$ induces in $\mathcal{G}_a$ a graph isomorphic to $G$. Whenever not stated differently in the paper, when we refer to a one-cluster automata, we assume that it is one-cluster with respect to a letter $a$.
Before moving further we state preliminary observation and lemma that can shed more light on the matter of synchronization and careful synchronization of one-cluster automata. 
\begin{observation}
\label{observation:cluster_defined}
If one-cluster automaton $\mathcal{A} = (Q, \Sigma, \delta)$ is carefully synchronizing, $C$ induces a directed cycle on a letter $a$, letter $a$ is the only one defined for all the states and $|C| > 1$, then there exist letter $b \in \Sigma$ such that $b \neq a$ and for any $q \in C$ it holds that $\delta(q, b)$ is defined.
\end{observation}
\begin{proof}
Obviously, since $a$ is the only letter defined for all the states, then every synchronizing word must start with it. On the other hand for all $k \in \mathbb{N}$ holds $C \subseteq Q.a^k$ and if $k \geq l$, where $l$ is the level of $\mathcal{G}_a$, then $Q.a^k = C$. If $|C| > 1$, then there must exits another letter $b \in \Sigma$ such that $C.b \neq C$ in order for $\mathcal{A}$ to be carefully synchronizing.
\end{proof}
Now we state and proof a lemma that shows, that if there exist a word $w$ that shrinks the set of the states that induces a directed cycle on $a$, then there exist a word $w'$ (polynomial in $|w|$) that shrinks the same subset by a half.
\begin{lemma}
\label{lemma:shrink_by_half}
Let $\mathcal{A} = (Q, \Sigma, \delta)$ be a one-cluster PFA, $|Q| = n$, let $C$ induce a directed cycle on a letter $a$ and $l$ be the level of $\mathcal{G}_a$. If there exist a word $w$ such that $|C.w| < |C|$ then there exist a word $w'$ such that $|Q.w'| \leq \lfloor \frac{1}{2} |C| \rfloor$ of length at most $n + \frac{1}{2}|C|(|w| + |C|)$.
\end{lemma}
\begin{proof}
For any $p, q \in C$ define $\text{dist}_{C}(p,q) = \text{min}\{k_1,k_2 : p.a^{k_1} = q \land q.a^{k_2} = p \land k_1, k_2 < |C|\}$. Since $p,q \in C$, then $dist_C(p,q)$ is always well defined and for all $p,q$ holds $dist_C(p,q) \leq \lfloor \frac{1}{2} |C| \rfloor$. Since there exist $w$ such that $|C| > |C.w|$, then there must exist $\bar{p}, \bar{q} \in C$ such that $\bar{p}.w = \bar{q}.w$. Denote $\text{dist}_C(\bar{p}, \bar{q}) = k$. We will show claim:\begin{claim}\label{claim_C_subset} For any $C' \subset C$ such that $|C'| > \frac{1}{2}|C|$, there exist $p', q' \in C$, such that $\text{dist}(p',q') = k$.\end{claim} 
\begin{proof}
For the sake of contradiction suppose that there exist $C' \subset C$ and $|C'| > \frac{1}{2}|C|$ such that for all $p,q \in C'$ holds $\text{dist}(p,q) \neq k$. Let $S = \{q_{i + k \text{ mod } m } : q_i \in C'\}$. Since for any $q_i \in C'$ there exist exactly one state $q_{i + k \text{ mod } m } \in C$, then $|C'| = |S|$. On the other hand, since for all $p,q \in C'$ holds $\text{dist}(p,q) \neq k$, then $C \cap S = \varnothing$. But since $|C'| = |S| > \frac{1}{2}|C|$ and $C' \subset C$ and $S \subset C$ then $|C'| + |S| > |C|$ so there must hold $C' \cap S \neq \varnothing$ and we have a contradiction.
\end{proof}
We can construct a word $w'$ in the following way:
Obviously $Q.a^l = C$ and $|Q.a^lw| < |C|$. If also $|Q.a^lw| \leq \lceil \frac{1}{2} |C| \rceil$ then result holds. Otherwise observe that $Q.a^lwa^l \subset C$ and $|Q.a^lwa^l| > \frac{1}{2} |C|$ so we can apply Claim \ref{claim_C_subset} to find states $p',q' \in Q.a^lwa^l$ such that $\text{dist}(p',q') = k = \text{dist}(\bar{p},\bar{q})$. Denote $u = a^lwa^l$ and observe that for some $m_1 \leq |C| < n$ it must hold that $\bar{p}, \bar{q} \in Q.ua^{m_1}$ so $|Q.ua^{m_1}w| < |Q.u|$. We can apply Claim \ref{claim_C_subset} and word $a^{m_1}w$ as long as the size of the resulting set is greater than $\frac{1}{2}|C|$. We obtain the word $w' = a^lw(a^{m_1}w)^m_2$ where $m_2 < \frac{1}{2}|C| - 1$ so the result holds.
\end{proof}

\section{Long shortest carefully synchronizing word}
\label{section:lower_bounds}
In this section we construct an infinite family of carefully synchronizing, one-cluster PFAs with exponential shortest carefully synchronizing word. The scheme of the proof is similar as in \cite{ito2004} (Proposition 8), however it differs in details. In order to simplify the proofs we assume, that the size of the set of states is $n = 2k$, where $k \in \mathbb{N}$ but proofs can be easily adapted to the case $n = 2k + 1$ achieving similar results. Let $C_k = \{c_1, \ldots, c_k\}$ and $T_k = \{t_1, \ldots , t_k\}$, $Q_k = C_k \cup  T_k$ and $\Sigma_k = \{a\} \cup \Sigma_k'$ where $\Sigma_k'$ is a set of letters specified later on. Finally let $\mathcal{B}_k = (Q_k, \Sigma_k, \delta_k)$. Define action of the letter $a$ on the set of states:
\begin{itemize}
    \item $\delta_k(c_i, a) = c_{i+1}$ for $i < k$
    \item $\delta_k(c_n, a) = c_1$
    \item $\delta_k(t_i, a) = c_i$
\end{itemize}
It is easy to observe that $\mathcal{B}_k$ is one-cluster with respect to $a$ for every $k$. Example of that cluster for $k=3$ is depicted below.
\begin{figure}[H]
    \centering
    \tikzset{every picture/.style={line width=0.75pt}} 

\begin{tikzpicture}[x=0.75pt,y=0.75pt,yscale=-1,xscale=1]

\draw    (271.38,198.59) -- (288.27,223.93) ;
\draw [shift={(289.38,225.59)}, rotate = 236.31] [color={rgb, 255:red, 0; green, 0; blue, 0 }  ][line width=0.75]    (10.93,-3.29) .. controls (6.95,-1.4) and (3.31,-0.3) .. (0,0) .. controls (3.31,0.3) and (6.95,1.4) .. (10.93,3.29)   ;
\draw  [fill={rgb, 255:red, 255; green, 255; blue, 255 }  ,fill opacity=1 ] (286.09,230.87) .. controls (288.13,226.54) and (293.29,224.68) .. (297.62,226.71) .. controls (301.95,228.75) and (303.81,233.91) .. (301.77,238.25) .. controls (299.74,242.58) and (294.57,244.44) .. (290.24,242.4) .. controls (285.91,240.36) and (284.05,235.2) .. (286.09,230.87) -- cycle ;
\draw    (265.17,146.15) -- (264.72,178.55) ;
\draw [shift={(264.69,180.55)}, rotate = 270.81] [color={rgb, 255:red, 0; green, 0; blue, 0 }  ][line width=0.75]    (10.93,-3.29) .. controls (6.95,-1.4) and (3.31,-0.3) .. (0,0) .. controls (3.31,0.3) and (6.95,1.4) .. (10.93,3.29)   ;
\draw  [fill={rgb, 255:red, 255; green, 255; blue, 255 }  ,fill opacity=1 ] (233.71,229.44) .. controls (235.75,225.11) and (240.91,223.25) .. (245.25,225.29) .. controls (249.58,227.32) and (251.44,232.49) .. (249.4,236.82) .. controls (247.36,241.15) and (242.2,243.01) .. (237.87,240.97) .. controls (233.54,238.94) and (231.68,233.77) .. (233.71,229.44) -- cycle ;
\draw  [fill={rgb, 255:red, 255; green, 255; blue, 255 }  ,fill opacity=1 ] (256.78,186.33) .. controls (258.82,182) and (263.98,180.14) .. (268.31,182.17) .. controls (272.64,184.21) and (274.5,189.37) .. (272.47,193.71) .. controls (270.43,198.04) and (265.27,199.9) .. (260.93,197.86) .. controls (256.6,195.82) and (254.74,190.66) .. (256.78,186.33) -- cycle ;
\draw  [fill={rgb, 255:red, 255; green, 255; blue, 255 }  ,fill opacity=1 ] (259.03,143.96) .. controls (255.26,141.01) and (254.59,135.57) .. (257.54,131.8) .. controls (260.49,128.02) and (265.94,127.36) .. (269.71,130.31) .. controls (273.48,133.26) and (274.14,138.7) .. (271.19,142.47) .. controls (268.24,146.24) and (262.8,146.91) .. (259.03,143.96) -- cycle ;
\draw    (337,256.67) -- (303.55,239.17) ;
\draw [shift={(301.77,238.25)}, rotate = 27.61] [color={rgb, 255:red, 0; green, 0; blue, 0 }  ][line width=0.75]    (10.93,-3.29) .. controls (6.95,-1.4) and (3.31,-0.3) .. (0,0) .. controls (3.31,0.3) and (6.95,1.4) .. (10.93,3.29)   ;
\draw  [fill={rgb, 255:red, 255; green, 255; blue, 255 }  ,fill opacity=1 ] (340.57,251.9) .. controls (344.86,249.77) and (350.06,251.52) .. (352.18,255.81) .. controls (354.31,260.1) and (352.56,265.3) .. (348.27,267.43) .. controls (343.98,269.55) and (338.78,267.8) .. (336.66,263.51) .. controls (334.53,259.22) and (336.28,254.02) .. (340.57,251.9) -- cycle ;
\draw    (200.73,253.06) -- (229.52,238.17) ;
\draw [shift={(231.3,237.26)}, rotate = 152.67] [color={rgb, 255:red, 0; green, 0; blue, 0 }  ][line width=0.75]    (10.93,-3.29) .. controls (6.95,-1.4) and (3.31,-0.3) .. (0,0) .. controls (3.31,0.3) and (6.95,1.4) .. (10.93,3.29)   ;
\draw  [fill={rgb, 255:red, 255; green, 255; blue, 255 }  ,fill opacity=1 ] (201.71,259.5) .. controls (200.88,264.22) and (196.4,267.38) .. (191.68,266.55) .. controls (186.96,265.73) and (183.81,261.24) .. (184.63,256.53) .. controls (185.45,251.81) and (189.94,248.66) .. (194.66,249.48) .. controls (199.37,250.3) and (202.53,254.79) .. (201.71,259.5) -- cycle ;
\draw    (249.4,236.82) -- (281.79,235.61) ;
\draw [shift={(283.78,235.53)}, rotate = 177.86] [color={rgb, 255:red, 0; green, 0; blue, 0 }  ][line width=0.75]    (10.93,-3.29) .. controls (6.95,-1.4) and (3.31,-0.3) .. (0,0) .. controls (3.31,0.3) and (6.95,1.4) .. (10.93,3.29)   ;
\draw    (245.25,225.29) -- (259.94,199.6) ;
\draw [shift={(260.93,197.86)}, rotate = 119.77] [color={rgb, 255:red, 0; green, 0; blue, 0 }  ][line width=0.75]    (10.93,-3.29) .. controls (6.95,-1.4) and (3.31,-0.3) .. (0,0) .. controls (3.31,0.3) and (6.95,1.4) .. (10.93,3.29)   ;

\draw (267.28,151.19) node [anchor=north west][inner sep=0.75pt]  [font=\footnotesize,rotate=-358.62] [align=left] {$\displaystyle a$};
\draw (257.94,238.76) node [anchor=north west][inner sep=0.75pt]  [font=\footnotesize,rotate=-357.36] [align=left] {$\displaystyle a$};
\draw (279.94,198.76) node [anchor=north west][inner sep=0.75pt]  [font=\footnotesize,rotate=-357.36] [align=left] {$\displaystyle a$};
\draw (238.94,204.76) node [anchor=north west][inner sep=0.75pt]  [font=\footnotesize,rotate=-357.36] [align=left] {$\displaystyle a$};
\draw (203.94,235.76) node [anchor=north west][inner sep=0.75pt]  [font=\footnotesize,rotate=-357.36] [align=left] {$\displaystyle a$};
\draw (322.94,235.76) node [anchor=north west][inner sep=0.75pt]  [font=\footnotesize,rotate=-357.36] [align=left] {$\displaystyle a$};
\draw (261.92,184.23) node [anchor=north west][inner sep=0.75pt]  [font=\tiny,rotate=-357.36] [align=left] {$\displaystyle c_{1}$};
\draw (237.94,227.76) node [anchor=north west][inner sep=0.75pt]  [font=\tiny,rotate=-357.36] [align=left] {$\displaystyle c_{2}$};
\draw (290.6,228.63) node [anchor=north west][inner sep=0.75pt]  [font=\tiny,rotate=-357.36] [align=left] {$\displaystyle c_{3}$};
\draw (258.91,132.96) node [anchor=north west][inner sep=0.75pt]  [font=\tiny,rotate=-359.21] [align=left] {$\displaystyle t_{1}$};
\draw (188.6,253.63) node [anchor=north west][inner sep=0.75pt]  [font=\tiny,rotate=-357.36] [align=left] {$\displaystyle t_{2}$};
\draw (342.71,254.8) node [anchor=north west][inner sep=0.75pt]  [font=\tiny,rotate=-357.36] [align=left] {$\displaystyle t_{3}$};

\end{tikzpicture}
    \caption{The $a$-cluster of the automaton $\mathcal{B}_3$}
    \label{fig:cluster}
\end{figure}
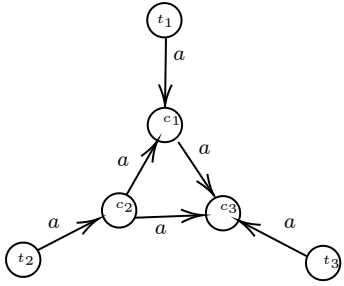
Let $\mathcal{T}_k = \{T \cup (C_k \cap ((T_k \setminus T).a).a^{-1}): T \in 2^{T_k} \setminus \{\varnothing\} \}$ Intuitively it is the set of all sets $T$ such that $|T'| = k$ and $T'.a = C_k$. First we prove a simple lemma.
\begin{lemma}
\label{lemma:T_k_size}
Let $\mathcal{T}_k$ be a family of sets defined above. Then $|\mathcal{T}_k| = 2^{|T_k|} - 1$.
\end{lemma}
\begin{proof}
Let $T'_1, T'_2 \in 2^{T_k}$ such that $T'_1 \neq T'_2$. Denote $C_k \cap ((T_k \setminus T'_1).a).a^{-1} = P_1$ and $C_k \cap ((T_k \setminus T'_2).a).a^{-1} = P_2$ It suffices to show that for any  $T'_1, T'_2$ sets $T'_1 \cup P_1$ and $T'_2 \cup P_2$ are different. But it is easy to notice that $P_1,P_2 \subset C_k$ and since $T'_1 \neq T'_2$ and $T'_1, T'_2 \subset T_k$ then the lemma holds.
\end{proof}
Using Lemma \ref{lemma:T_k_size} we can enumerate sets of family $\mathcal{T}_k$ arbitrarily from $S_1$ to $S_{2^{|T_k|} - 1}$. Let us prove properties of that sets in the next lemma.
\begin{lemma}
\label{lemma:S_i_sets}
For any $i \in \{1, \ldots, 2^{|T_k|} - 1 \}$ holds $S_i.a = C_k$ and $|S_i| = \frac{n}{2}$.
\end{lemma}
\begin{proof}
By definition of $S_i$ we know that there exist $T \subseteq T_k$ such that $S_i = T \cup (C_k \cap ((T_k \setminus T).a).a^{-1})$. By construction $T.a \subseteq C$ so it suffices to show that $(C \cap ((T_k \setminus T).a).a^{-1}).a  = C \setminus T.a$. Indeed $(C_k \cap ((T_k \setminus T).a).a^{-1}).a  = C_k.a \cap (T_k \setminus T).a = C_k \cap T_k.a \setminus T.a = C \cap C \setminus T.a = C \setminus T.a$. To show that $|S_i| = \frac{n}{2}$ first observe that  $(T_k \setminus T).a = \{ q \in C_k \setminus T.a\}$. Let $q \in C_k \setminus T.a$ and observe that set $q.a^{-1}$ contains exactly two states $q_1 \in T_k$ and $q_2 \in C_k$ and by the construction we have that $|((T_k \setminus T).a).a^{-1} | = 2|C_k \setminus T.a|$ and half of the states must be contained in $C_k$ which gives us $|S_i| =|T.a| + |C_k \setminus T.a|$ and that concludes the proof.
\end{proof}
Let $\Sigma_k' = \{b_1, \ldots, b_{2^{|T_k|} -1}, c\}$ (this is the subalphabet used in the definition of $\Sigma_k$) and $S_i = \{s_1^i, \ldots, s_k^i\}$. Put transition function $\delta_k$ (for $\Sigma_k'$) as follows:
\begin{itemize}
    \item $\delta_k(c_j, b_1) = s_j^1$, letter $b_1$ not defined for other states
    \item for $0 < i < 2^{|T_k|} -2$ let $\delta_k(s_j^i, b_{i+1}) = s_j^{i+1}$, letter $b_i+1$ not defined for other states
    \item$\delta_k(s_j^{2^{|T_k|}-1}, c) = c_1$, letter $c$ not defined for other states
\end{itemize}
Observe that for any $i$ holds $\delta_k(S_i,b_{i+1}) = S_{i+1}$. Let us state and prove the main theorem in this section.
\begin{theorem}
\label{theorem:lower_bound}
Automaton $\mathcal{B}_k$ is carefully synchronized by a word $w = ab_1\ldots b_{2^{|T_k|}-1}c$, $w$ is the shortest word that carefully synchronizes $\mathcal{B}_k$ and $|w| = 2^{\frac{n}{2}} + 1$.
\end{theorem}
\begin{proof}
Since $\mathcal{B}_k$ is one clustered with respect to $a$ and the level of that cluster is $1$, it is straightforward that $Q.a = C$. Immediately from the construction we have that for $w_i = ab_1\ldots b_i$ the equality $Q.w_i = S_i$ holds and from that we have that $Q.w = c_1$ so $w$ carefully synchronizes $\mathcal{B}_k$. In order to prove the minimality of $w$ first notice that for each $i$ holds that $S_i.b_j$ is undefined for $j \neq i-1$ (because each $|S_i| = \frac{n}{2}$ from Lemma \ref{lemma:S_i_sets}) and $S_i.a = C$ (by Lemma \ref{lemma:S_i_sets}). Then $\mathcal{P}(\mathcal{B}_k)$ forms a path from $Q$ to $\{c_1\}$ labelled with consecutive letters of $w$ and for each state on the path there is only one transition leading to $C$ (visited in the first step of the path) and that ends the proof.
\end{proof}

\section{Complexity}
\label{section:complexity}

This section is devoted to the proof that the problem of deciding whether a given one-cluster PFA is carefully synchronizing is NP-hard even in the case of binary alphabet. Notice that Observation \ref{observation:cluster_defined} implies that, letter $b$ must be defined for at least all states in the cycle induced by the letter $a$. Refer to that problem as ONE-CLUSTER-CARSYNC. Let us state that theorem formally.
\begin{theorem}
\label{theorem:np_hardness}
For a given PFA $\mathcal{A} = (Q, \{a,b\}, \delta)$ that is one-cluster with respect to $a$, the problem of deciding whether $\mathcal{A}$ has a carefully synchronizing word (2-ONE-CLUSTER-CARSYNC) is NP-hard.
\end{theorem}
In order to prove the theorem we construct a polynomial time reduction from 3-SAT to 2-ONE-CLUSTER-CARSYNC. Let $\{x_1, \ldots, x_n\}$ be a set of variables and $C_1, \ldots, C_m$ be clauses. We also assume that for any pair $\{x_j, \bar{x}_j\}$ if $x_j \in C_i$ then $\bar{x}_j \notin C_i$. That is not really a problem, since if both of them belong to $C_i$, then $C_i$ is always true. For a given formula $\phi = C_1 \land \ldots \land C_m$ we define a PFA $\mathcal{A}_{\phi} = (Q, \{a,b\}, \delta_{\phi})$. Let $Q = \bigcup_{i=1}^m (S_{C_i^t} \cup S_{C_i^f}) \cup P \cup R$ where:
\begin{itemize}
    \item $P = \{p_1, \ldots, p_m\}$
    \item $R = \{r_1, \ldots, r_m\}$
    \item $S_{C_i^t} = \{c_i^{x_1}, \ldots, c_i^{x_n}, x_i^{end}\}$ for each $i$
    \item $S_{C_i^f} = \{\bar{c}_i^{x_1}, \ldots, \bar{c}_i^{x_n}, \bar{x}_i^{end}\}$ for each $i$.
\end{itemize}
Define $\delta_{\phi}$ in the following way:
\begin{enumerate}
    \item \label{point:1} $\delta_{\phi}(p_i, a) = p_{i+1\text{ mod }m}$
    \item \label{point:2} $\delta_{\phi}(p_i, b) = \bar{c}_i^{x_1}$
    \item \label{point:3} $\delta_{\phi}(r_i, a) = p_{i}$, $\delta_{\phi}(r_i, b)$ undefined
    \item \label{point:4} $\delta_{\phi}(c_i^{end}, a) = \delta_{\phi}(\bar{c}_i^{end}, a) = r_i$
    \item \label{point:5} $\delta_{\phi}(c_i^{end}, b) = p_1$ $\delta_{\phi}(\bar{c}_i^{end},b)$ undefined
    \item \label{point:6} if $x_j \in C_i$ then $\delta_{\phi}(\bar{c}_i^{x_j}, a) = c_i^{x_{j+1}}$ otherwise $\delta_{\phi}(\bar{c}_i^{x_j}, a) = \bar{c}_i^{x_{j+1}}$ for $0 < j < n$
    \item \label{point:7} if $\bar{x}_j \in C_i$ then $\delta_{\phi}(\bar{c}_i^{x_j}, b) = c_i^{x_{j+1}}$ otherwise $\delta_{\phi}(\bar{c}_i^{x_j}, b) = \bar{c}_i^{x_{j+1}}$ for $0 < j < n$
    \item \label{point:8} $\delta_{\phi}(c_i^{x_j}, a)= \delta_{\phi}(c_i^{x_j}, b) = c_i^{x_{j+1}}$ for $0 < j < n$ 
    \item \label{point:9} if $x_n \in C_i$ then $\delta_{\phi}(\bar{c}_i^{x_n}, a) = c_i^{end}$ otherwise $\delta_{\phi}(\bar{c}_i^{x_n}, a) = \bar{c}_i^{end}$
    \item \label{point:10} if $\bar{x}_n \in C_i$ then $\delta_{\phi}(\bar{c}_i^{x_n}, b) = c_i^{end}$ otherwise $\delta_{\phi}(\bar{c}_i^{x_n}, b) = \bar{c}_i^{end}$
    \item \label{point:11} $\delta_{\phi}(c_i^{x_n}, a)= \delta_{\phi}(c_i^{x_n}, b) = c_i^{end}$ 
\end{enumerate}
For the example formula $\phi_{ex} = (x_1 \lor x_3 \lor x_4) \land (x_1 \lor x_2 \lor \bar{x}_3) \land (\bar{x}_1 \lor \bar{x}_2 \lor x_4)$ the corresponding $\mathcal{A}_{\phi_{ex}}$ is depicted in Fig. \ref{fig:example_phi_ex}.
\begin{figure}[H]
    \centering
    \input{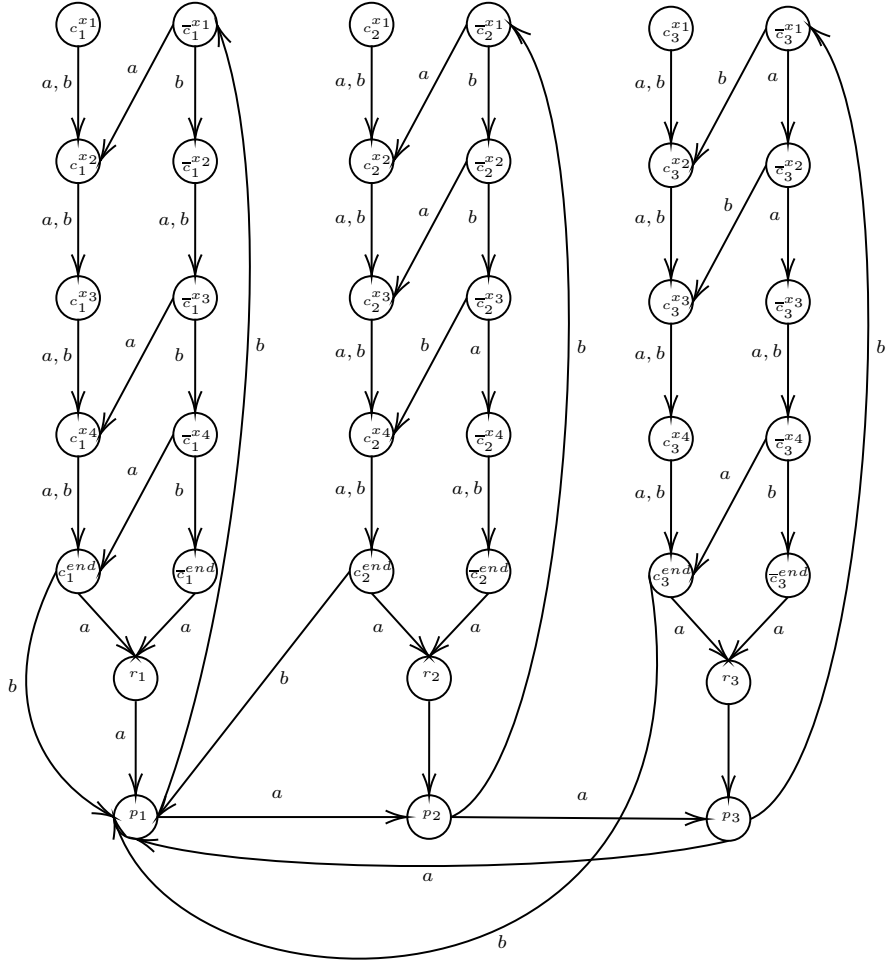}
    \caption{Automaton $\mathcal{A}_{\phi_{ex}}$}
    \label{fig:example_phi_ex}
\end{figure}
Let us state and prove two lemmas before going further.
\begin{lemma}
\label{lemma:a_phi_cluster}
For any given $\phi$ the automaton $\mathcal{A}_{\phi}$ is one-cluster with respect to $a$.
\end{lemma}
\begin{proof}
First observe that for any $\phi$ the set $P$ forms the $a$-cycle in the automaton $\mathcal{A}_\phi$ (Point \ref{point:1} of the definition of $\delta_{\phi}$). It is also easy to see from the definition of $\mathcal{A}_{\phi}$ that any set $S_{C_i^t}$ together with the states $r_i \in R$ and $p_i \in P$ creates a directed path of the form $c_i^{x_1} \rightarrow c_i^{x_2} \rightarrow \ldots \rightarrow  c_i^{x_n} \rightarrow r_i \rightarrow p_i$ labelled with $a$ (Points \ref{point:3}, \ref{point:4} and \ref{point:8} of the definition of $\delta_{\phi}$). Denote this path as $p$. Denote by $v_j$ a variable of the formula $\phi$ and $v_j \in \{x_j, \bar{x}_j\}$. Consider now clause $C_i \in \phi$ and the set $S_{C_i^f}$. Assume that $C_i = (v_{j_1} \lor v_{j_2} \lor v_{j_3})$ and $j_1 < j_2 < j_3$. Any of the variables $v_{j_3 + 1}, \ldots ,v_n$ does not belong to clause $C_i$, so the set $\{\bar{c}_i^{x_{j_3 + 1}}, \bar{c}_i^{x_{j_3 + 2}}, \ldots, \bar{c}_i^{x_{n}}, \bar{c}_i^{end}, r_i\}$ forms a directed path on the letter $a$ (Points \ref{point:4}, \ref{point:6}, \ref{point:7}, \ref{point:9}, \ref{point:10} of the definition of $\delta_{\phi}$). Denote it by $p_1$. Suppose that $v_{j_3} = \bar{x}_{j_3}$. Then, by Points \ref{point:6} or \ref{point:9}, state $\bar{c}_i^{x_{j_3}}$ is attached to the first state of the path $p_1$ with a transition labelled by a letter $a$. Otherwise, by Points \ref{point:6} or \ref{point:9}, $\bar{c}_i^{x_{j_3}}$ is attached to the path $p$ by a letter $a$. Now we can repeat our argument to the set $\{\bar{c}_i^{x_{j_2 + 1}}, \ldots  \bar{c}_i^{x_{j_3 - 1}}\}$ to show that it forms a path $p_2$ labelled with the letter $a$ which begins in the state $\bar{c}_i^{x_{j_2 + 1}}$ and ends in the state $\bar{c}_i^{x_{j_3 - 1}}$. By Point \ref{point:6} we obtain, that this path is attached by its end to the state $\bar{c}_i^{x_{j_3}}$ by the letter $a$. Now we can also repeat our arguments to the state $\bar{c}_i^{x_{j_2 + 1}}$ to show that it is either attached to $p_2$ or to $p$. Using the same reasoning for the sets $\{\bar{c}_i^{x_{j_1 + 1}}, \ldots  \bar{c}_i^{x_{j_2 - 1}}\}$ and $\{\bar{c}_i^{x_1}, \ldots  \bar{c}_i^{x_{j_2 - 1}}\}$ and state $\bar{c}_i^{x_{j_1}}$ concludes the proof.
\end{proof}

\begin{lemma}
\label{lemma:a_phi_d_undefinded}
For any given $\phi$, if $n_1 < n + 3$ then $\delta_{\phi}(Q,a^{n_1}b)$ is undefined.
\end{lemma}
\begin{proof}
Observe, that the letter $a$ induces a path on $n + 3$ vertices on the set $S_{C_i^t} \cup \{r_i,p_i\} = S_i$ for each $i$. That path starts in $c_i^{x_1}$, ends in $p_i$ and the one before last vertex on that path is $r_i$. Also $Q = \bigcup_{i=1}^m S_i$, so for each $n_1 < n + 3$ holds $r_i \in Q.a^{n_1}$ for any $0 < i < m+1$. On the other hand $\delta_\phi(r_i,b)$ is undefined for each $i$ (Point \ref{point:3}), so $\delta_{\phi}(Q,a^{n_1}b)$ is undefined and the lemma holds. 
\end{proof}
Before moving further let us define, for a given $\mathcal{A}_{\phi}$, two subsets of its states $S_{init} = \{\bar{c}_1^{x_1}, \ldots, \bar{c}_m^{x_1}\}$ and $S_{end} = \{c_1^{end}, \ldots, c_m^{end}\}$ and state the following observation.
\begin{observation}
\label{observation:n_length_word}
For any $i \in \{1, \ldots, m\}$ and formula $\phi$, if $v \in \{a,b\}^*$ and $|v| = n$, then $\delta_{\phi}(\bar{c}_i^{x_1}, v) \in \{c_i^{end},\bar{c}_i^{end}\}$.
\end{observation}
\begin{proof}
Immediate from Points \ref{point:6} to \ref{point:11} of definition of $\delta_{\phi}$.
\end{proof}
\begin{lemma}
\label{lemma:equiv_word}
For any given $\phi$ automaton $\mathcal{A}_{\phi}$ is carefully synchronizing if and only if there exist a word $w$ of length $n$ such that $\delta_{\phi}(S_{init},w) = S_{end}$.
\end{lemma}
\begin{proof}
First assume, that there exists a word $w$ of length $n$ such that $\delta_{\phi}(S_{init},w) = S_{end}$. Observe that $\delta_{\phi}(Q,a^{n+3}) = P$ and $\delta_{\phi}(P,b) = S_{init}$. Also $\delta_{\phi}(S_{init},w) = S_{end}$ and $\delta_{\phi}(S_{end},b) = p_1$ so we conclude that word $u = a^{n+3}bwb$ carefully synchronizes $\mathcal{A}_{\phi}$. Now assume that $\mathcal{A}_{\phi}$ is carefully synchronizing. From Lemma \ref{lemma:a_phi_d_undefinded} we obtain that any carefully synchronizing word for $\mathcal{A}_{\phi}$, say $u \in \{a,b\}^*$, must start with a word $a^{k}$ where $k > n+2$. Since $\delta_{\phi}(Q, a^{n+3}) = P$ and $\delta_{\phi}(P, a) = P$ we imply that $u$ must be of the form $a^kbv$. We know that $\delta_{\phi}(Q, a^kb) = S_{init}$.  From Observation \ref{observation:n_length_word} for any $v_1 \in \{a,b\}^*$ such that $|v_1| = n$ we have $\delta_{\phi}(\bar{c}_i^{x_1},v_1) \in \{c_i^{end},\bar{c}_i^{end}\}$. On the other hand observe that for any set of the form $\{t_1, \ldots, t_n\}$, where $t_i \in \{c_i^{end},\bar{c}_i^{end}\}$ but $S_{end}$, we have that $\delta_\phi(T, b)$ is undefined, $\delta_{\phi}(T,a) = R$ and $\delta_{\phi}(R,b)$ is undefined, and $\delta_{\phi}(R,a) = P$. So, since $u$ carefully synchronizes $a$ there must exist desired word, which concludes the proof.
\end{proof}
\begin{lemma}
\label{lemma:equiv_formula}
Let $\phi$ be a formula and $\mathcal{A}_{\phi}$ be a corresponding automaton. Then $\phi$ has truth assignment if and only if there exists a word $w$ of length $n$ such that $\delta_{\phi}(S_{init},w) = S_{end}$. 
\end{lemma}
\begin{proof}
First assume that $\phi$ has a truth assignment $e: \{x_1, \ldots, x_n\} \rightarrow \{0, 1\}$. We define a word $w$ of length $n$ in a following way: if $e(x_i) = 1$ in that evaluation then $i$-th letter of the word $w$ is $a$, otherwise the $i$-th letter is $b$. Note $w_i$ as $i-1$ letter prefix of $w$ and $\bar{w}_i$ as a $i-1$ letter suffix of $w$. Let $0 < k_1 <k_2 < k_3 < n+1$ and consider clause $C_j = y_{k_1} \lor y_{k_2} \lor y_{k_3}$ where $y_{k_i} \in \{x_{k_i}, \bar{x}_{k_i}\}$. Since $e$ is a truth evaluation, then at least one of the variables $x_{k_1}$, $x_{k_2}$ or $x_{k_3}$ must evaluate to true in the sense that if $y_{k_i} = x_{k_i}$ then $e(x_{k_i}) = 1$, otherwise $e(x_{k_i}) = 0$. It is straightforward from the definition of $\delta_{\phi}$ (Points \ref{point:6} and \ref{point:7}) that $\delta_{\phi}(\bar{c}_j^{x_1}, w_{k_1}) = \bar{c}_j^{x_{k_1}}$. Now consider four pairwise exclusive cases:
\begin{enumerate}
    \item \label{case:1} $y_{k_1} = x_{k_1}$ and $e(x_{k_1}) = 1$
    \item \label{case:2} $y_{k_1} = x_{k_1}$ and $e(x_{k_1}) = 0$
    \item \label{case:3} $y_{k_1} = \bar{x}_{k_1}$ and $e(x_{k_1}) = 1$
    \item \label{case:4} $y_{k_1} = \bar{x}_{k_1}$ and $e(x_{k_1}) = 0$.
\end{enumerate}
Now:
\begin{itemize}
    \item if case \ref{case:1} holds, then $i$-th letter of $w$ is $a$ and by Point \ref{point:6} we know that $\delta_{\phi}(\bar{c}_j^{x_{k_1}}, a) = c_j^{x_{k_1+1}}$
    \item if case \ref{case:2} holds, then $i$-th letter of $w$ is $b$ and by Point \ref{point:7} we know that $\delta_{\phi}(\bar{c}_j^{x_{k_1}}, b) = \bar{c}_j^{x_{k_1+1}}$
    \item if case \ref{case:3} holds, then $i$-th letter of $w$ is $a$ and by Point \ref{point:6} we know that $\delta_{\phi}(\bar{c}_j^{x_{k_1}}, a) = \bar{c}_j^{x_{k_1+1}}$
    \item if case \ref{case:4} holds, then $i$-th letter of $w$ is $b$ and by Point \ref{point:7} we know that $\delta_{\phi}(\bar{c}_j^{x_{k_1}}, b) = c_j^{x_{k_1+1}}$.
\end{itemize}
If case 1 or 4 hold, then $\delta_{\phi}(\bar{c}_j^{x_1}, w_{i+1}) = c_j^{x_{k_1+1}}$ and, from Points \ref{point:8} and \ref{point:11} we obtain that $\delta_{\phi}(c_j^{x_{k_1+1}}, \bar{w}_{n-k_1}) = c_j^{end}$. Otherwise we obtain that $\delta(\bar{c}_j^{x_1}, w_{k_2}) = \delta(\bar{c}_j^{x_{k_2}})$ and we can repeat our case analysis to obtain that either $\delta(\bar{c}_j^{x_1}, w_{k_2 + 1}) = c_j^{x_{k_2+1}}$ (if $y_{k_2}$ evaluates to true) or  $\delta(\bar{c}_j^{x_1}, w_{k_2 + 1}) = \bar{c}_j^{x_{k_2+1}}$ (if $y_{k_2}$ evaluates to false). With the same argument applied to $y_{k_3}$, since at least one of the variables of $C_j$ must be evaluated to true, we obtain that $\delta_{\phi}(\bar{c}_j^{x_1}, w) = c_j^{end}$. But since $e$ is a truth evaluation, then any $C_j$ must have at least one variable evaluated to true, so for any $j \in \{1, \ldots, m\}$ we have that $\delta_{\phi}(\bar{c}_j^{x_1},w) = c_j^{end}$, so we obtain that $\delta(S_{init}, w) = S_{end}$.\newline 
Now assume that there exist a word $w$ of length $n$, such that $\delta(S_{init}, w) = S_{end}$. We define the evaluation of $\phi$, $e_w: \{x_1, \ldots, x_n\} \rightarrow \{0,1\}$ in the following way: if $i$-th letter of $w$ is $a$, then $e_w(x_i) = 1$ otherwise $e_w(x_i) = 0$. Since $\delta_{\phi}(S_{init}, w) = S_{end}$, then from Observation \ref{observation:n_length_word} for any $j \in \{1, \ldots, m\}$ holds $\delta_{\phi}(\bar{c}_j^{x_1}, w) = c_j^{end}$. Consider evaluation $e_w$ and let $0 < k_1 < 
k_2 < k_3 < n+1$ and clause $C_j = y_{k_1} \lor y_{k_2} \lor y_{k_3}$ where $y_{k_i} \in \{x_{k_i}, \bar{x}_{k_i}\}$. It suffices to prove that for any $j$ if $\delta_{\phi}(\bar{c}_j^{x_1}, w) = c_j^{end}$, then $e_w(y_{k_1}) = 1$ or $e_w(y_{k_1}) = 1$ or $e_w(y_{k_1}) = 1$ (we assume here that $e(\bar{x}) = 1 - e(x)$). For the sake of contradiction suppose, that there exist $j \in \{1,\ldots, m\}$ such that $\delta_{\phi}(\bar{c}_j^{x_1}, w) = c_j^{end}$ and $e_w(y_{k_1}) = 0$ and $e_w(y_{k_2}) = 0$ and $e_w(y_{k_3}) = 0$. First observe that $\delta_{\phi}(\bar{c}_j^{x_1}, w_{k_1}) = \bar{c}_{j}^{x_{k_1}}$. Since $e(y_{k_1}) = 0$, then (from Points \ref{point:6} and \ref{point:7}) we obtain that $\delta_{\phi}(\bar{c}_j^{x_1}, w_{k_1 + 1}) = \bar{c}_{j}^{x_{k_1 + 1}}$. The same argument for $y_2$ gives us $\delta_{\phi}(\bar{c}_j^{x_1}, w_{k_2 + 1}) = \bar{c}_{j}^{x_{k_2 + 1}}$ (from Points \ref{point:6} and \ref{point:7}) and further for $y_3$ $\delta_{\phi}(\bar{c}_j^{x_1}, w_{k_3 + 1}) = \bar{c}_{j}^{x_{k_3 + 1}}$ or $\delta_{\phi}(\bar{c}_j^{x_1}, w_{k_3 + 1}) = \bar{c}_{j}^{end}$ (respectively from Points \ref{point:6} and \ref{point:7} or \ref{point:10} and \ref{point:11}). In the latter case we obtain contradiction straightforward and in the former we notice, that for any $v \in \{a,b\}^*$ such that $|v| = n - k_3$ we obtain that $\delta_{\phi}(\bar{c}_j^{x_{k_3}}, v) = \bar{c}_{j}^{end}$ (from Points \ref{point:6}, \ref{point:7},\ref{point:10} and \ref{point:11} and observation that any variable with index greater than $k_3$ does not belong to the clause $C_j$). That concludes the proof.
\end{proof}
Combining Lemmas \ref{lemma:equiv_word} and \ref{lemma:equiv_formula} we obtain that automaton $\mathcal{A}_{\phi}$ is carefully synchronizing if, and only if formula $\phi$ has a truth evaluation. Taking Lemma \ref{lemma:a_phi_cluster} into account, the only thing we have to prove in order to finish the proof of Theorem \ref{theorem:np_hardness} is that the reduction can be performed in polynomial time. Observe that, for any $\phi$ with $n$ variables and $m$ clauses and corresponding $\mathcal{A}_{\phi}$ we have that $|P| = |R| = m$ and for any $i \in \{1, \ldots, m\}$ holds also $|S_{C_i^t}| = |S_{C_i^f}| = n+1$, we obtain, that $|Q| = 2m(n+2)$ and on the other hand $\mathcal{A}_{\phi}$ is a partial automaton, what means, that for any state there are at most two outgoing transitions ($a$ and $b$), so computing automaton $\mathcal{A}_{\phi}$ can be done in polynomial size of $n$ and $m$ and that concludes the proof of Theorem \ref{theorem:np_hardness}.
\section{Conclusions}
In Section \ref{section:lower_bounds} we have found infinite family of PFAs with shortest carefully synchronizing words of exponential length. However the size of the alphabet used in the construction is also exponential in terms of the number of states and that makes the construction not sufficient to make analysis of decision problem as stated in Section \ref{section:complexity}. An interesting problem to investigate is if one can achieve exponentially long shortest carefully synchronizing word for one-cluster PFAs, but using smaller alphabet size.\newline
In Section \ref{section:complexity} we have proven NP-hardness of the problem of deciding whether a given binary one-cluster PFAs can be carefully synchronized. Let $\mathcal{A}$ be a one cluster PFA with respect to a letter $a$ and let $C$ be a set of states that induces a cycle on letter $a$. Further analysis of the proof of NP-hardness leads to a conclusion that even the problem of deciding whether there exist a word $w$ such that $|C.w| < |C|$ is also NP-hard. Let us make a simple observation:
\begin{observation}
\label{observation:witness}
Let $\mathcal{A} = (Q, \Sigma, \delta)$ be a PFA, and $w \in \Sigma^*$. There exist an algorithm that, in $O(|w||Q|)$ time complexity decides whether $w$ is a carefully synchronizing word for $\mathcal{A}$.
\end{observation}
\begin{proof}
Consider algorithm defined below:
\begin{algorithmic}
\State $P \gets Q$
\For{$i=0,i < |w|,i \gets i+1$} 
    \State $P \gets P.w[i]$
\EndFor
\If{$|P| = 1$}
\State \Return true
\Else
\State \Return false
\EndIf
\end{algorithmic}
Obviously algorithm answers true if and only if $w$ carefully synchronizes $\mathcal{A}$ and requires $O(|Q|)$ additional space complexity and $O(|Q||w|)$ time complexity, what concludes the proof.
\end{proof}
From Observation \ref{observation:witness} we can infer that one possibility to prove that the problem 2-ONE-CLUSTER-CARSYNC is in NP is to prove that the shortest carefully synchronizing word for binary one-cluster automata is of polynomial length, a problem that we what we want to investigate in further research. Lemma \ref{lemma:shrink_by_half} can be utilized in the proof of that result. One must show first that if there exist a word $w$ such that $|C.w| < C$ where $C$ is the set of states inducing a cycle on the letter $a$, then such a word is of polynomial length in terms of the set of states. Next question is how to synchronize the remaining half of the states of the cycle.

\end{document}